\documentclass[journal,comsoc,draftcls,onecolumn,12pt,twoside]{IEEEtran}

\usepackage[T1]{fontenc}

\ifCLASSINFOpdf
  \usepackage[pdftex]{graphicx}
\else
  \usepackage[dvips]{graphicx}
\fi
\usepackage[cmex10]{amsmath}
\usepackage{cite}
\usepackage[utf8]{inputenc}
\usepackage{amsthm}
\usepackage{multirow}
\newtheorem{theorem}{Theorem}
\newtheorem{lemma}{Lemma}

\newtheorem{definition}{Definition}
\newtheorem{remark}{Remark}

\DeclareMathOperator{\tr}{\mathrm{tr}}
\DeclareMathOperator{\SNR}{\mathrm{SNR}}
\def\BibTeX{{\rm B\kern-.05em{\sc i\kern-.025em b}\kern-.08em
    T\kern-.1667em\lower.7ex\hbox{E}\kern-.125emX}}

\newcommand{\sunderline}[1]{\mkern3mu\underline{\mkern-3mu#1\mkern-3mu}\mkern3mu }
\newcommand{\soverline}[1]{\mkern2mu\overline{\mkern-2mu#1\mkern-2mu}\mkern2mu }

\begin{document}
%
\title{Optimality of the Proper Gaussian Signal in Complex MIMO Wiretap Channels}
%
%
%

\author{Yong~Dong,~\IEEEmembership{Student Member,~IEEE,}
        Yinfei~Xu,~\IEEEmembership{Member,~IEEE,}
        Tong~Zhang,~\IEEEmembership{Member,~IEEE,}
        and~Yili~Xia,~\IEEEmembership{Member,~IEEE}
\thanks{Yong~Dong, Yinfei~Xu, and Yili~Xia are with the School of Information Science and Engineering, Southeast University, Nanjing 210096, China (e-mail: dyjjd@seu.edu.cn; yinfeixu@seu.edu.cn; yili-xia@seu.edu.cn). Tong~Zhang is with the Department of Electrical and Electronic Engineering, Southern University of Science and Technology, Shenzhen 518055, China (e-mail: zhangt7@sustech.edu.cn).

Part of this paper has been submitted to WCNC 2023.}
}

%
%

\markboth{Journal of \LaTeX\ Class Files,~Vol.~14, No.~8, August~2015}%
{Dong \MakeLowercase{\textit{et al.}}: MIMO Wiretap Channel}
%



\maketitle

\begin{abstract}
The multiple-input multiple-output (MIMO) wiretap channel (WTC), which has a transmitter, a legitimate user and an eavesdropper, is a classic model for studying information theoretic secrecy. In this paper, the fundamental problem for the complex WTC is whether the proper signal is optimal has yet to be given explicit proof, though previous work implicitly assumed the complex signal was proper. Thus, a determinant inequality is proposed to prove that the secrecy rate of a complex Gaussian signal with a fixed covariance matrix in a degraded complex WTC is maximized if and only if the signal is proper, i.e. the pseudo-covariance matrix is a zero matrix. Moreover, based on the result of the degraded complex WTC and the min-max reformulation of the secrecy capacity, the optimality of the proper signal in the general complex WTC is also revealed. The results of this research complement the current research on complex WTC. To be more specific, we have shown it is sufficient to focus on the proper signal when studying the secrecy capacity of the complex WTC.
\end{abstract}


%
\IEEEpeerreviewmaketitle

\section{Introduction}
The broadcast nature of wireless communications has identified secure communications as an important issue for many years. Traditionally, communication security is achieved through cryptographic techniques to achieve computational security at an upper layer. However, whether a communication system is computational secure is decided by the comparison of the complexity of attacking the system and the computing power of the attacker, which is under certain conditions. In contrast, physical security at the lowest level was proposed to ensure secure communications regardless of the attacker's computing power \cite{shannonCommunicationTheorySecrecy1949}. 
\subsection{Wiretap Channel and Proper Signal}
\begin{figure}[!t]
\centering
\includegraphics[width=0.7\textwidth]{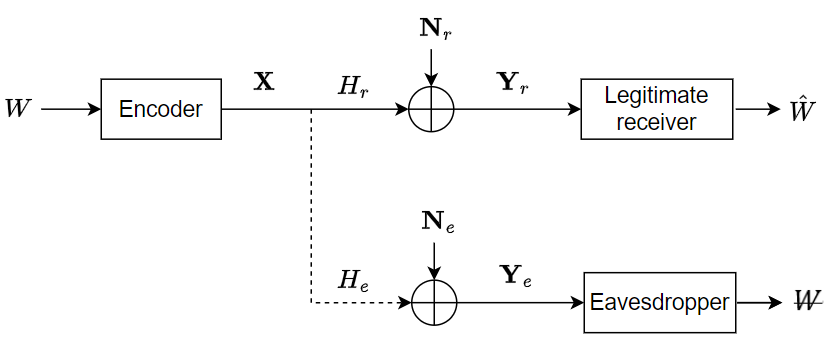}
\caption{MIMO WTC, with $n_t$, $n_r$, and $n_e$ antennas respectively at the transmitter, legitimate receiver, and eavesdropper.}
\label{fig_model}
\end{figure}
 A well-known model of physical layer security is the MIMO WTC \cite{leung-yan-cheongGaussianWiretapChannel1978}, in which a transmitter (Alice) wishes to communicate with a legitimate receiver (Bob) in the presence of an eavesdropper (Eve), as shown in Fig. \ref{fig_model}. Denote the number of antennas for Alice, Bob and Eve by $n_t$, $n_r$, and $n_e$ respectively. Let $H_r \in \mathbb{C}^{n_r \times n_t}$ and $H_e \in \mathbb{C}^{n_e \times n_t}$ be the channel matrices for the legitimate user and eavesdropper. Then the signals received at Bob and Eve can be formulated as
\begin{subequations}\label{eq:1}
\begin{align}
\mathbf{Y}_r &=H_r\mathbf{X}+\mathbf{N}_r, \label{eq:1a}\\
\mathbf{Y}_e &= H_e\mathbf{X}+\mathbf{N}_e, \label{eq:1b}
\end{align}
\end{subequations}
 where $\mathbf{N}_r \in \mathbb{C}^{n_r \times 1} $ and $\mathbf{N}_e \in \mathbb{C}^{n_e \times 1}$ are independent proper Gaussian white noise, i.e. $\mathbf{N}_r  \sim \mathcal{CN}(\mathbf{0},I_{n_r})$ and $\mathbf{N}_e  \sim \mathcal{CN}(\mathbf{0},I_{n_e})$.
For the zero-mean complex random vector $\mathbf{X}$, its covariance matrix and pseudo-covariance matrix are denoted by
\begin{equation*}
    K_{\mathbf{X}}=\mathbb{E}\{\mathbf{X}\mathbf{X}^H\},\quad \tilde{K}_{\mathbf{X}}=\mathbb{E}\{\mathbf{X}\mathbf{X}^T\},
\end{equation*} 
respectively. $K_{\mathbf{X}}$ is Hermitian and positive semidefinite, while $\tilde{K}_{\mathbf{X}}$ is symmetric. If the pseudo-covariance matrix $\tilde{K}_{\mathbf{X}}=0$, $\mathbf{X}$ turns to a proper Gaussian signal; otherwise $\mathbf{X}$ is called improper. 
Besides, for a pair of Hermitian and positive semidefinite matrix $K$ and symmetric matrix $\tilde{K}$, there exists a random vector with covariance and pseudo-covariance given by $K$ and $\tilde{K}$ if and only if (iff) the augmented covariance matrix $\sunderline{K}$ satisfies
\begin{equation}\label{eq:2}
    \sunderline{K} \triangleq 
    \begin{bmatrix}
    K & \tilde{K}\\
    \tilde{K}^* & K^*
    \end{bmatrix}
    \succeq 0,
\end{equation}
i.e., $\sunderline{K}$ is positive semidefinite\cite[Section 2.2.2]{schreierStatisticalSignalProcessing2010}.
\subsection{The Secrecy capacity of the complex WTC with proper signals}
The term secrecy capacity, which was originally described as the maximum rate of reliable transmission under perfect secrecy, was first introduced by Wyner, as shown in the following definition.
\begin{definition}[Secrecy Capacity]
The perfect secrecy capacity $C_s$ is the maximum achievable rate $R_e$ that makes the decoding error at the legitimate receiver and the information leakage at the eavesdropper both tend to zero\cite{wynerWireTapChannel1975}. 
\end{definition}
A single letter expression of the secrecy capacity of the discrete memoryless (DM) WTC with transition probability $p(y_r,y_e|x)$ is given in \cite{csiszarBroadcastChannelsConfidential1978} as
\begin{equation}\label{eq:3}
    C_s=\mathop{\max}_{p(u,x)} [I(U;Y_r)-I(U;Y_e)],
\end{equation}
in which the auxiliary random variable $U$ and signal variables $X,Y_r$ and $Y_e$ form a Markov chain $U \to X \to (Y_r,Y_e)$, where $|\mathcal{U}| \leq |\mathcal{X}|$. As for the MIMO complex WTC under a sum power constraint, it has been proved that the secrecy capacity is
\begin{align}\label{eq:4}
    C_p&=\mathop{\max}_{K_{\mathbf{X}}\in \mathcal{K} } R_p(K_{\mathbf{X}})=  \mathop{\max}_{K_{\mathbf{X}}\in \mathcal{K} } \left[ \log\det(I_{n_r}+H_r K_{\mathbf{X}} H_r^H)-\log\det(I_{n_e}+H_e K_{\mathbf{X}} H_e^H) \right],
\end{align}
where $\mathcal{K}=\{K_{\mathbf{X}}|K_{\mathbf{X}}\succeq 0, \tr(K_{\mathbf{X}})  \leq P\}$ is the set of all possible covariance matrices and $R_p(K_{\mathbf{X}})$ is the maximum achievable secrecy rate between Alice and Bob when the covariance matrix of the transmitted signal is $K_{\mathbf{X}}$ \cite{khistiSecureTransmissionMultiple2010,khistiSecureTransmissionMultiple2010a,oggierSecrecyCapacityMIMO2011,tieliuNoteSecrecyCapacity2009}.

Nevertheless, \eqref{eq:4} is actually the secrecy capacity with proper transmitted signals rather than general complex signals. Previous studies did not explicitly explain why the optimal Gaussian signal is proper, which is not evident from our point of view. The research would have been more complete if the secrecy capacity of the complex WTC with general complex signals had been given or if the optimality of the proper signal had been demonstrated.
\subsection{Literature Review and Motivations}
This subsection is devoted to an overview of the improper signal and the WTC. For a general complex-valued signal, it is either proper or improper \cite{schreierStatisticalSignalProcessing2010}. Over the past few decades, it has usually been implicitly assumed that the complex signal is proper, which means the complex signal is uncorrelated with its complex conjugate. When calculating the channel capacity, treating complex signals as proper signals is convenient because the related calculation and the results are similar to those of real signals. Besides, Neeser et al. reveal that the differential entropy of a complex random vector with a fixed correlation matrix is maximized if and only if (iff) the vectore is proper, Gaussian and zero mean \cite{neeserProperComplexRandom1993}, so the proper signal is capacity achieving for the point-to-point channel \cite{telatarCapacityMultiantennaGaussian1999a} and multiple access channel (MAC) \cite[Lemma 1]{hellingsOptimalityProperSignaling2014}. In addition, it was shown in \cite{vishwanathDualityAchievableRates2003a} that the Sato upper bound on the sum rate of the Gaussian MIMO broadcast channel (BC) can be achieved by dirty paper coding (DPC) with proper Gaussian signals under a sum power constraint. Furthermore, in \cite{hellingsOptimalityProperSignaling2014}, the opitimality of proper signals is generalized to Gaussian MIMO BC with DPC under a sum covariance constraint and it points out that the user encodes at last in a BC channel is without interference, so the link from the transmitter to this user is similar to a point-to-point channel, hence proper signals are optimal for the last link, which leads to the optimality of proper signals for the penultimate link and all the other links so forth. As a result, it gradually becomes a common assumption in the theoretical analysis of information theory that the transmitted signal is proper as long as the capacity achieving distribution proved to be Gaussian.

However, there are other situations where improper Gaussian signals outperform proper Gaussian signals with respect to channel capacities or achievable rates. Most of them are interference-dominant situations because the existence of interference makes the achievable rate in the form of a difference of two logarithm-determinant functions (or equivalently a logarithm-determinant function with a non-constant denominator), such as the interference channel (IC) \cite{hoImproperGaussianSignaling2012,zengTransmitOptimizationImproper2013,zengOptimizedTransmissionImproper2013,kurniawanImproperGaussianSignaling2015,lameiroRateRegionBoundary2017}, the relay channel \cite{gaafarImproperGaussianSignaling2016,gaafarAlternateRelayingImproper2016} and the cognitive network \cite{lameiroBenefitsImproperSignaling2015,aminUnderlayCognitiveRadio2016,lameiroMaximallyImproperSignaling2019}. Nevertheless, Neeser et al.'s result appears to only guarantee that one logarithm-determinant function or the sum of several logarithm-determinant functions can be maximized with proper signals \cite{neeserProperComplexRandom1993}. In \cite{zengTransmitOptimizationImproper2013}, it was shown that the use of improper signals brings about an improvement of achievable rates over proper signals in ICs with interference treated as additive noise. In \cite[Remark 1]{zengOptimizedTransmissionImproper2013}, it described heuristically that from the user k's perspective in the IC, if interference is treated as noise, the user k essentially communicates over a point-to-point channel. However, the proper signal is strictly sub-optimal in this point-to-point channel, which is mainly because the noise containing interference can be improper if the transmitted signal is improper, so the analysis for the original point-to-point channel is no longer applicable here. Moreover, the superiority of improper signals was shown and a closed form real-composite transmit covariance matrix was derived to maximize the sum-rate of the Z-interference channel in \cite{kurniawanImproperGaussianSignaling2015}. With the help of the circularity coefficient, \cite{lameiroRateRegionBoundary2017} obtained a necessary and sufficient condition to decide when the improper signal is optimal for the Z-IC. For the two-user BC with two private messages, it has been derived that every boundary point is achieved when at least one user employs the improper signal in \cite{lameiroImproperGaussianSignaling2019}. 

In this paper, we consider complex Gaussian signals for the WTC. Ever since introduced by Wyner in \cite{wynerWireTapChannel1975} and extended to Gaussian case in \cite{leung-yan-cheongGaussianWiretapChannel1978}, the secrecy capacity of the MIMO Gaussian WTC has yet to be given in closed form. Nevertheless the MISO WTC \cite{khistiSecureTransmissionMultiple2010,liOptimalInputCovariance2010} and some special MIMO WTCs, such as the strictly degraded Gaussian MIMO WTC with sufficiently large power \cite{loykaOptimalSignalingSecure2016}, were studied and given with the optimal transmitted covariance matrix. In addition, by assuming the optimal covariance matrix is full-rank, a method was presented for characterizing the optimal covariance matrix with an arbitrary number of antennas \cite{fakoorianFullRankSolutions2013}. Recently, a series of optimization method has been put forward to find the numerical result of the Gaussian MIMO WTC. In \cite{cumananSecrecyRateOptimizations2014}, a difference of convex functions algormithm (DCA) was proposed for the secrecy rate of the WTC. The secrecy rate is with the form of a two logarithm-determinant functions thus is the difference of two convex functions. While another numerical method is based on the min-max (convex-concave) reformulation of the secrecy capacity for the Gaussian MIMO WTC \cite{nguyenLowComplexityAlgorithmAchieving2020}. More recently, an accelerated algorithm (ACDA) and a partial best response algorithm (PBRA) are proposed in \cite{mukherjeeSecrecyCapacityMIMO2021}. 

These works focus on proper signals, but we observe that although from the legitimate user's perspective the communication is point-to-point, the existence of an eavesdropper makes the secrecy capacity of the MIMO WTC also has the form of the difference of two logarithm-determinant functions, which is similar to the achievable rate for one specific user in the IC, where the improper signal can enlarge the achievable rate. On the other hand, we cannot directly assert that improper signals are optimal since improper signals sometimes may bring no gains for the IC \cite{hellingsImproperSignalingTimeSharing2017}. Therefore, in this paper, we try to decide whether the proper signal is optimal for the WTC. If it is capacity achieving, we need a new method to prove its optimality because the result of point-to-point channel in \cite{telatarCapacityMultiantennaGaussian1999a} does not apply for the WTC. This motivates us to construct a determinant inequality which functions the same as the Fischer inequality, which is used in proving proper Gaussian zero-mean signal can maximize the differential entropy \cite[Result 2.2]{schreierStatisticalSignalProcessing2010}. With this inequality, we prove that the proper signal is optimal in achieving the secrecy capacity of the complex WTC.

\subsection{Main Contributions}
This paper establishes several results about the secrecy capacity of the complex MIMO WTC. 
\begin{itemize}
\item This paper relies on information measures' dependence on the distribution of random variables rather than the form of random variables to obtain the secrecy rate of a general complex signal,  which is expressed with augmented covariance matrix. Besides, this paper characterizes the secrecy capacity with general complex signals under a sum power constraint, which is a maximization of achievable rates over the set of all possible augmented covariance matrices. Thereafter we briefly demonstrate that the previous achievable rate for a proper signal is a special case of the achievable rate for a general complex signal by choosing a zero pseudo-covariance matrix. 
\item This paper derives a matrix determinant inequality from a simple conditional entropy inequality. This determinant inequality, which is similar to the Fischer-inequality, plays an important role in proving that the proper signal is capacity achieving in the degraded WTC.
\item The determinant inequality suggests that the secrecy rate of a complex random vector with a fixed correlation matrix is maximized if and only if the signal is proper, Gaussian and zero-mean, where the optimalilty of Gaussian and zero-mean signals has been proven in the previous work and the superiority of proper signals over improper signals is demonstrated with the determinant inequality in this work. As a result, the proper signal is capacity achieving for the degraded complex WTC.
\item  This paper utilizes the min-max and max-min reformulation of the secrecy capacity of the WTC, which is the minimization on the set of correlation matrices of two augmented noises and the maximization on the set of augmented covariance matrices. Then the result of the degraded WTC can be applied to the min-max reformulated expression and presents the secrecy capacity is achieved when the signal is proper. 
\end{itemize}
\subsection{Paper Organization and Notation}
The rest of this paper is organized as follows. Section \ref{sec2} studies the secrecy capacity of the general WTC, which is expressed with the augmented covariance matrix. Section \ref{sec3} focuses on the degraded WTC and proves that the proper signal can achieve the secrecy capacity. In Section \ref{sec4}, the optimality of the proper signal for the general WTC is proved based on the result of degraded WTC and the min-max reformulation of the secrecy capacity. In section \ref{sec5}, we utilizes two algorithms in previous research to compare the maximum achievable rate for the proper signal and general complex signal. Finally, we conclude the paper in Section \ref{sec6}.

We use lowercase letters $x,y,\cdots$ to denote scalars. Meanwhile, we use uppercase letters $X,Y,\cdots$ and bold uppercase letters $\mathbf{X},\mathbf{Y},\cdots$ to denote random variables and random vectors, respectively. Calligraphic letters $\mathcal{X},\mathcal{Y},\cdots$ are used for finite sets, and $\dim(\mathcal{X})$ stands for the cardinality of the set $\mathcal{X}$. We also denote the dimension of random vector $\mathbf{X}$ by $\dim(\mathbf{X})$. Matrices are written in uppercase letters $A,B,\cdots$ and $\tr(\cdot)$, $\det(\cdot)$, $(\cdot)^*$, $(\cdot)^T$, and $(\cdot)^H$ are used to denote the trace, determinant, complex-conjugate, transpose, and conjugate transpose of a matrix, respectively. Besides, $\mathcal{CN}(\boldsymbol{\mu}, K)$ denotes a circularly symmetric complex-valued Gaussian random vector with mean vector $\boldsymbol{\mu}$ and covariance $K$. $I_m$ represents the identity matrix of size m. $\mathbb{E}\{\cdot\}$ denotes the expectation of a random variable. $\mathfrak{R}\{\cdot\}$ and $\mathfrak{I}\{\cdot\}$ represent the real and imaginary parts of a complex number, respectively.

\section{The Secrecy Capacity of the complex WTC with General Complex Signals}\label{sec2}
In the section that follows, the situation that the transmitted signal is a general complex signal will be argued. In contrast to the assumption in the previous study that the transmitted signal is a proper signal, the secrecy rate for general complex signals is expressed with augmented covariance matrices and augmented channel matrices. Although any complex system may be transformed into an equivalent real system, the elegance of the original system description will lose \cite{schreierSecondorderAnalysisImproper2003}. Therefore, we will rely on the equivalent real WTC to obtain the secrecy rate of the complex WTC, but we rewrite the secrecy rate with complex matrices from the complex system so as to keep the elegance of the expression, hence the secrecy rate of the WTC with general complex signal is similar to the secrecy rate of the WTC with the proper signal.
\begin{theorem}\label{thm:1}
The secrecy capacity of the MIMO WTC with a general complex signal is
\begin{equation}\label{eq:5}
    C_g=\mathop{\max}_{\sunderline{K}_{\mathbf{X}}\in \mathcal{\sunderline{K}}_g} R_g(\sunderline{K}_{\mathbf{X}})=\mathop{\max}_{\sunderline{K}_{\mathbf{X}}\in \mathcal{\sunderline{K}}_g}\left[\frac{1}{2}\log\det(I_{2n_r}+\sunderline{H}_r\sunderline{K}_{\mathbf{X}}\sunderline{H}_r^H)-\frac{1}{2}\log\det(I_{2n_e}+\sunderline{H}_e\sunderline{K}_{\mathbf{X}}\sunderline{H}_e^H)\right].
\end{equation}
The augmented covariance $\sunderline{K}_{\mathbf{X}}$ is defined in \eqref{eq:2} and the augmented channel matrices $\sunderline{H}_r$ and $\sunderline{H}_e$ are defined as
\begin{equation}
    \sunderline{H}_r=\begin{bmatrix}
    H_r & 0\\
    0 & H_r^*
    \end{bmatrix},\quad\sunderline{H}_e=\begin{bmatrix}
    H_e & 0\\
    0 & H_e^*
    \end{bmatrix}.
\end{equation}
 The feasible set $\mathcal{\sunderline{K}}_g=\{\sunderline{K}_\mathbf{X}|\sunderline{K}_{\mathbf{X}}=\begin{bmatrix}
 K_\mathbf{X} & \tilde{K}_\mathbf{X}\\
\tilde{K}_\mathbf{X}^* & K_\mathbf{X}^*
 \end{bmatrix} \succeq 0, \tr(\sunderline{K}_{\mathbf{X}}) \leq 2P\}$ is the set of all possible augmented covariance matrices.
\end{theorem}
\begin{proof}
Please refer to Appendix \ref{appendixa}.
\end{proof}

Note that if transmitted signals are proper, which means $\tilde{K}_{\mathbf{X}}=0$, we will denote the set of proper signals' augmented covariance matrices by $\mathcal{\sunderline{K}}_p$ as
\begin{align}
\mathcal{\sunderline{K}}_p=\{\sunderline{K}_\mathbf{X}|\sunderline{K}_{\mathbf{X}}=\begin{bmatrix}
 K_\mathbf{X} & 0\\
0 & K_\mathbf{X}^*
 \end{bmatrix} \succeq 0, \tr(\sunderline{K}_{\mathbf{X}}) \leq 2P\},
\end{align}
then we have
\begin{align}
C_g&=\mathop{\max}_{\sunderline{K}_{\mathbf{X}}\in \mathcal{\sunderline{K}}_p} \left[ \frac{1}{2}\log\det(I_{2n_r}+\sunderline{H}_r\sunderline{K}_{\mathbf{X}}\sunderline{H}_r^H)-\frac{1}{2}\log\det(I_{2n_e}+\sunderline{H}_e\sunderline{K}_{\mathbf{X}}\sunderline{H}_e^H) \right]\nonumber\\
&=\mathop{\max}_{K_{\mathbf{X}}\in \mathcal{K}} \left[ \frac{1}{2}\log\det\left(I_{2n_r}+\begin{bmatrix}
    H_r & 0\\
    0 & H_r^*
    \end{bmatrix}
\begin{bmatrix}
    K_\mathbf{X} & 0\\
    0 & K_\mathbf{X}^*
\end{bmatrix}
\begin{bmatrix}
    H_r & 0\\
    0 & H_r^*
\end{bmatrix}^H\right)\right.\nonumber\\
&\left.\quad-\frac{1}{2}\log\det\left(I_{2n_e}+\begin{bmatrix}
    H_e & 0\\
    0 & H_e^*
    \end{bmatrix}
\begin{bmatrix}
    K_\mathbf{X} & 0\\
    0 & K_\mathbf{X}^*
\end{bmatrix}
\begin{bmatrix}
    H_e & 0\\
    0 & H_e^*
\end{bmatrix}^H\right) \right]\nonumber\\
&=\mathop{\max}_{K_{\mathbf{X}}\in \mathcal{K}} \left[ \log\det(I_{n_r}+H_r K_{\mathbf{X}} H_r^H)-\log\det(I_{n_e}+H_e K_{\mathbf{X}} H_e^H) \right]\nonumber\\
&=C_p
\end{align}
which means \eqref{eq:5} is the same as \eqref{eq:4} if the transmitted signal is proper. However, if the transmitted signal is improper, then \eqref{eq:5} will be different to \eqref{eq:4}. In fact, because $\mathcal{\sunderline{K}}_p \subseteq \mathcal{\sunderline{K}}_g$, we find
\begin{align}\label{eq:9}
C_p\leq C_g
\end{align}
always holds.
Although proper Gaussian signals are capacity achieving in many channels, improper Gaussian signals achieve larger rate regions in interference channels when the achievable rate is the difference of logarithm-determinant functions. In this paper, the optimality of the proper Gaussian signal will be proved using a Fischer-like determinant inequality. As a result, the secrecy capacity in \eqref{eq:4} is also the secrecy capacity of the WTC with general complex signals.
\section{Optimality of Proper Signal for the Degraded MIMO WTC}\label{sec3}
If the channel is degraded , i.e. $H_r^H H^r-H_e^H H_e \succ 0$\cite[Section 1.C]{oggierSecrecyCapacityMIMO2011}, then let
\begin{equation}\label{eq:delta}
    \Delta=H_r^H H_r-H_e^H H_e \succ 0,
\end{equation}
and let
\begin{equation}\label{eq:delta_}
    \underline{\Delta}=\sunderline{H}_r^H\sunderline{H}_r-\sunderline{H}_e^H\sunderline{H}_e=\begin{bmatrix}
    \Delta & 0\\
    0 & \Delta^*\\
    \end{bmatrix}.
\end{equation}
$\Delta$ and $\sunderline{\Delta}$ are definite, which means there exist $\Delta^\frac{1}{2}$ and $\underline{\Delta}^\frac{1}{2}$\footnote[2]{$\Delta^\frac{1}{2}$ is any matrix $A$ satisfying $A^HA=\Delta$. If unitary diagonalization is used to construct $A$ then $A=A^H.$}. In order to show the optimality of proper signals in degraded WTCs, firstly we will list several matrix equalities that will be used. 
\begin{itemize}
\item 
For any $A\in\mathbb{C}^{m\times n}$ and $B \in \mathbb{C}^{n \times m}$, we have\footnote{Sketch of proof: 
Let $
M=\begin{bmatrix}
I_m & -A\\B & I_n
\end{bmatrix},
G=\begin{bmatrix}
I_m & A\\ 0 & I_n
\end{bmatrix},
$
then we have $\det(MG)=\det(GM)$.}
\begin{align} \label{eq:m1}
\det(I_m+AB)=\det(I_n+BA).
\end{align}

\item
For any $A\in\mathbb{C}^{p\times n}$, $B \in \mathbb{C}^{q \times n}$ and $C \in \mathbb{C}^{n \times n}$, we have
\begin{align}\label{eq:m2}
\begin{bmatrix}
A\\B
\end{bmatrix}
C
\begin{bmatrix}
A^H&B^H
\end{bmatrix}
=\begin{bmatrix}
A&0\\0&B
\end{bmatrix}
\begin{bmatrix}
C&C\\C&C
\end{bmatrix}
\begin{bmatrix}
A^H&0\\0&B^H
\end{bmatrix}.
\end{align}
\item
Since elementary row operation and column operation do not change the determinant of a matrix, for a block matrix
\begin{align}
A=\begin{bmatrix}
A_{11}&A_{12}&A_{13}&A_{14}\\
A_{21}&A_{22}&A_{23}&A_{24}\\
A_{31}&A_{32}&A_{33}&A_{34}\\
A_{41}&A_{42}&A_{43}&A_{44}\\
\end{bmatrix},
\end{align}
we can interchange the second row with the third row and interchange the second column with the third column at the same time and won't change the value of the determinant, thus we can obtain
\begin{align}\label{eq:swap}
\det(A)=\det\left(\begin{bmatrix}
A_{11}&A_{13}&A_{12}&A_{14}\\
A_{31}&A_{33}&A_{32}&A_{34}\\
A_{21}&A_{23}&A_{22}&A_{24}\\
A_{41}&A_{43}&A_{42}&A_{44}\\
\end{bmatrix}\right).
\end{align}
\end{itemize}

In addition to above matrix equalities, an determinant inequality, which is similar to Fischer determinant inequality, is put forward and will play a key role in the proof. Let $A$ be a square matrix of size $n$. $\mathcal{S}_1, \mathcal{S}_2 \in \{1,2,\cdots,n\}$ are two index sets. We denote by $A(\mathcal{S}_1,\mathcal{S}_2)$ the submatrix of entries that lie in the rows indexed by $\mathcal{S}_1$ and columns indexed by $\mathcal{S}_2$. If $\mathcal{S}_1=\mathcal{S}_2$, $A(\mathcal{S}_1)=A(\mathcal{S}_1,\mathcal{S}_1)$ is called a principal submatrix of $A$.
\begin{lemma}\label{lemma:1}
Let $K$ be a positive definite $k \times k$ matrix. Index sets  $\mathcal{S}_1,\mathcal{S}_2,\mathcal{S}_3$ and $\mathcal{S}_4$, which have the form of $\mathcal{S}_1=[1:k_1],\mathcal{S}_2=[k_1+1:k_2],\mathcal{S}_3=[k_2+1:k_3]$ and $\mathcal{S}_4=[k_3+1:k]$, constitute a sequential partition of $\{1,2,\cdots,k\}$. Then we have
\begin{equation}\label{eq:fish}
    \frac{\det(K)}{\det(K(\mathcal{S}_2 \cup \mathcal{S}_4))} \leq \frac{\det(K(\mathcal{S}_1 \cup \mathcal{S}_2))}{\det(K(\mathcal{S}_2))} \cdot \frac{\det(K(\mathcal{S}_3 \cup \mathcal{S}_4))}{\det(K(\mathcal{S}_4))},
\end{equation}
with equality iff $K(\mathcal{S}_1 \cup \mathcal{S}_2,\mathcal{S}_3 \cup \mathcal{S}_4)=K(\mathcal{S}_3 \cup \mathcal{S}_4,\mathcal{S}_1 \cup \mathcal{S}_2)=0$.
\end{lemma}
\begin{proof}
Please refer to Appendix \ref{appendixb}.
\end{proof}
\begin{theorem}
The secrecy capacity of degraded complex WTC is achieved when the signal is proper.
\end{theorem}
\begin{proof}
Let $\mathbf{X}$ be a general Gaussian complex signal with covariance matrix $K_\mathbf{X}$ and pseudo covariance matrix $\tilde{K}_\mathbf{X}$, then for its secrecy rate, we have
\begin{align}
&\qquad \frac{1}{2}\log\det(I_{2n_r}+\sunderline{H}_r\sunderline{K}_{\mathbf{X}}\sunderline{H}_r^H)-\frac{1}{2}\log\det(I_{2n_e}+\sunderline{H}_e\sunderline{K}_{\mathbf{X}}\sunderline{H}_e^H) \nonumber\\
&\overset{(a)}{=} \frac{1}{2}\log\det(I_{2n_t}+\sunderline{H}_r^H\sunderline{H}_r\sunderline{K}_{\mathbf{X}})-\frac{1}{2}\log\det(I_{2n_e}+\sunderline{H}_e\sunderline{K}_{\mathbf{X}}\sunderline{H}_e^H) \nonumber\\
&\overset{(b)}{=} \frac{1}{2}\log\det(I_{2n_t}+(\sunderline{\Delta}+\sunderline{H}_e^H\sunderline{H}_e)\sunderline{K}_{\mathbf{X}})-\frac{1}{2}\log\det(I_{2n_e}+\sunderline{H}_e\sunderline{K}_{\mathbf{X}}\sunderline{H}_e^H) \nonumber\\
&=\frac{1}{2}\log\det\left(I_{2n_t}+\begin{bmatrix}\underline{\Delta}^\frac{H}{2}&\sunderline{H}_e^H\end{bmatrix}\begin{bmatrix}\underline{\Delta}^\frac{1}{2}\\\sunderline{H}_e\end{bmatrix}\sunderline{K}_{\mathbf{X}}\right)-\frac{1}{2}\log\det(I_{2n_e}+\sunderline{H}_e\sunderline{K}_{\mathbf{X}}\sunderline{H}_e^H) \nonumber\\
&\overset{(a)}{=} \frac{1}{2}\log\det\left(I_{2(n_t+n_e)}+\begin{bmatrix}\underline{\Delta}^\frac{1}{2}\\\sunderline{H}_e\end{bmatrix}\sunderline{K}_{\mathbf{X}}\begin{bmatrix}\underline{\Delta}^\frac{H}{2}&\sunderline{H}_e^H\end{bmatrix}\right)-\frac{1}{2}\log\det(I_{2n_e}+\sunderline{H}_e\sunderline{K}_{\mathbf{X}}\sunderline{H}_e^H) \nonumber\\
&\overset{(c)}{=} \frac{1}{2}\log\det\left(I_{2(n_t+n_e)}+\begin{bmatrix}\underline{\Delta}^\frac{1}{2}&0\\0&\sunderline{H}_e\end{bmatrix}\begin{bmatrix}\sunderline{K}_{\mathbf{X}}&\sunderline{K}_{\mathbf{X}}\\\sunderline{K}_{\mathbf{X}}&\sunderline{K}_{\mathbf{X}}\end{bmatrix}\begin{bmatrix}\underline{\Delta}^\frac{H}{2}&0\\0&\sunderline{H}_e^H\end{bmatrix}\right)\nonumber \\
&\qquad-\frac{1}{2}\log\det(I_{2n_e}+\sunderline{H}_e\sunderline{K}_{\mathbf{X}}\sunderline{H}_e^H) \nonumber\\
&=\frac{1}{2}\log\det\left( I_{2(n_t+n_e)}+\begin{bmatrix}\Delta^\frac{1}{2}&0&0&0\\0&(\Delta^\frac{1}{2})^*&0&0\\0&0&H_e&0\\0&0&0&H_e^*\end{bmatrix}\begin{bmatrix}K_{\mathbf{X}}&\tilde{K}_{\mathbf{X}}&K_{\mathbf{X}}&\tilde{K}_{\mathbf{X}}\\\tilde{K}_{\mathbf{X}}^*&K_{\mathbf{X}}^*&\tilde{K}_{\mathbf{X}}^*&K_{\mathbf{X}}^*\\K_{\mathbf{X}}&\tilde{K}_{\mathbf{X}}&K_{\mathbf{X}}&\tilde{K}_{\mathbf{X}}\\\tilde{K}_{\mathbf{X}}^*&K_{\mathbf{X}}^*&\tilde{K}_{\mathbf{X}}^*&K_{\mathbf{X}}^*\end{bmatrix}\begin{bmatrix}\Delta^\frac{H}{2}&0&0&0\\0&\Delta^\frac{T}{2}&0&0\\0&0&H_e^H&0\\0&0&0&H_e^T\end{bmatrix} \right)\nonumber\\
 &\qquad-\frac{1}{2}\log\det\left(I_{2n_e}+\begin{bmatrix} H_e&0\\0&H_e^*\end{bmatrix}
 \begin{bmatrix}
     K_\mathbf{X} & \tilde{K}_\mathbf{X}\\
    \tilde{K}_\mathbf{X}^* & K_\mathbf{X}^*
 \end{bmatrix} 
 \begin{bmatrix}H_e^H&0\\0&H_e^T \end{bmatrix}
 \right)\nonumber\\
&\overset{(d)}{=}\frac{1}{2}\log\det\left( I_{2(n_t+n_e)}+\begin{bmatrix}\Delta^\frac{1}{2}&0&0&0\\0&H_e&0&0\\0&0&(\Delta^\frac{1}{2})^*&0\\0&0&0&H_e^*\end{bmatrix}
 \begin{bmatrix}
 K_{\mathbf{X}}&K_{\mathbf{X}}&\tilde{K}_{\mathbf{X}}&\tilde{K}_{\mathbf{X}}\\
 K_{\mathbf{X}}&K_{\mathbf{X}}&\tilde{K}_{\mathbf{X}}&\tilde{K}_{\mathbf{X}}\\
 \tilde{K}_{\mathbf{X}}^*&\tilde{K}_{\mathbf{X}}^*&K_{\mathbf{X}}^*&K_{\mathbf{X}}^*\\
 \tilde{K}_{\mathbf{X}}^*&\tilde{K}_{\mathbf{X}}^*&K_{\mathbf{X}}^*&K_{\mathbf{X}}^*
 \end{bmatrix}\begin{bmatrix}\Delta^\frac{H}{2}&0&0&0\\0&H_e^H&0&0\\0&0&\Delta^\frac{T}{2}&0\\0&0&0&H_e^T\end{bmatrix} \right)\nonumber\\
 &\qquad-\frac{1}{2}\log\det\left(I_{2n_e}+\begin{bmatrix} H_e&0\\0&H_e^*\end{bmatrix}
 \begin{bmatrix}
     K_\mathbf{X} & \tilde{K}_\mathbf{X}\\
    \tilde{K}_\mathbf{X}^* & K_\mathbf{X}^*
 \end{bmatrix} 
 \begin{bmatrix}H_e^H&0\\0&H_e^T \end{bmatrix}
 \right),
\end{align} 
where (a) follows by \eqref{eq:m1}, (b) follows by \eqref{eq:delta_}, (c) follows by \eqref{eq:m2} and (d) follows by \eqref{eq:swap}.
Next, denote $K$ as 
\begin{align}
K=I_{2(n_t+n_e)}+\begin{bmatrix}\Delta^\frac{1}{2}&0&0&0\\0&H_e&0&0\\0&0&(\Delta^\frac{1}{2})^*&0\\0&0&0&H_e^*\end{bmatrix}
 \begin{bmatrix}
 K_{\mathbf{X}}&K_{\mathbf{X}}&\tilde{K}_{\mathbf{X}}&\tilde{K}_{\mathbf{X}}\\
 K_{\mathbf{X}}&K_{\mathbf{X}}&\tilde{K}_{\mathbf{X}}&\tilde{K}_{\mathbf{X}}\\
 \tilde{K}_{\mathbf{X}}^*&\tilde{K}_{\mathbf{X}}^*&K_{\mathbf{X}}^*&K_{\mathbf{X}}^*\\
 \tilde{K}_{\mathbf{X}}^*&\tilde{K}_{\mathbf{X}}^*&K_{\mathbf{X}}^*&K_{\mathbf{X}}^*
 \end{bmatrix}\begin{bmatrix}\Delta^\frac{H}{2}&0&0&0\\0&H_e^H&0&0\\0&0&\Delta^\frac{T}{2}&0\\0&0&0&H_e^T\end{bmatrix},
\end{align}
and divide $K$'s index set $[1:2(n_t+n_e)]$ into $\mathcal{S}_1=[1:n_t],\mathcal{S}_2=[n_t+1:n_t+n_e],\mathcal{S}_3=[n_t+n_e+1:2n_t+n_e]$ and $\mathcal{S}_4=[2n_t+n_e+1:2(n_t+n_e)]$, then we have
\begin{align}
&\quad\frac{1}{2}\log\det\left( I_{2(n_t+n_e)}+\begin{bmatrix}\Delta^\frac{1}{2}&0&0&0\\0&H_e&0&0\\0&0&(\Delta^\frac{1}{2})^*&0\\0&0&0&H_e^*\end{bmatrix}
 \begin{bmatrix}
 K_{\mathbf{X}}&K_{\mathbf{X}}&\tilde{K}_{\mathbf{X}}&\tilde{K}_{\mathbf{X}}\\
 K_{\mathbf{X}}&K_{\mathbf{X}}&\tilde{K}_{\mathbf{X}}&\tilde{K}_{\mathbf{X}}\\
 \tilde{K}_{\mathbf{X}}^*&\tilde{K}_{\mathbf{X}}^*&K_{\mathbf{X}}^*&K_{\mathbf{X}}^*\\
 \tilde{K}_{\mathbf{X}}^*&\tilde{K}_{\mathbf{X}}^*&K_{\mathbf{X}}^*&K_{\mathbf{X}}^*
 \end{bmatrix}\begin{bmatrix}\Delta^\frac{H}{2}&0&0&0\\0&H_e^H&0&0\\0&0&\Delta^\frac{T}{2}&0\\0&0&0&H_e^T\end{bmatrix} \right)\nonumber\\
 &\qquad-\frac{1}{2}\log\det\left(I_{2n_e}+\begin{bmatrix} H_e&0\\0&H_e^*\end{bmatrix}
 \begin{bmatrix}
     K_\mathbf{X} & \tilde{K}_\mathbf{X}\\
    \tilde{K}_\mathbf{X}^* & K_\mathbf{X}^*
 \end{bmatrix} 
 \begin{bmatrix}H_e^H&0\\0&H_e^T \end{bmatrix}
 \right)\nonumber\\
 &=\frac{1}{2}\log\det(K)-\frac{1}{2}\log\det(K(\mathcal{S}_2 \cup \mathcal{S}_4))]\nonumber\\
 &\overset{(a)}{\leq}  \frac{1}{2}\log\det(K(\mathcal{S}_1 \cup \mathcal{S}_2)K(\mathcal{S}_3 \cup \mathcal{S}_4))-\frac{1}{2}\log\det(K(\mathcal{S}_2)K(\mathcal{S}_4))\nonumber\\
 &\overset{(b)}{=}\log\det(K(\mathcal{S}_1 \cup \mathcal{S}_2))-\log\det(K(\mathcal{S}_2))\nonumber\\
 &=\log\det\left(I_{n_r+n_t}+
\begin{bmatrix}
\Delta^\frac{1}{2}&0\\0&H_e
\end{bmatrix}
\begin{bmatrix}
K_\mathbf{X}&K_\mathbf{X}\\K_\mathbf{X}&K_\mathbf{X}
\end{bmatrix}
\begin{bmatrix}
\Delta^\frac{H}{2}&0\\0&H_e^H
\end{bmatrix}\right)-\log\det\left(I_{n_e}+H_e K_\mathbf{X}H_e^H\right),
\end{align}
where (a) follows by \eqref{eq:fish} and (b) follows since $K(\mathcal{S}_1 \cup \mathcal{S}_2)$ is the transpose of $K(\mathcal{S}_3 \cup \mathcal{S}_4)$ and $K(\mathcal{S}_2)$ is the transpose of $K(\mathcal{S}_4)$, so we have $\det(K(\mathcal{S}_1 \cup \mathcal{S}_2)=\det(K(\mathcal{S}_3 \cup \mathcal{S}_4))$ and $\det(K(\mathcal{S}_2))=\det(K(\mathcal{S}_4))$. 
Finally, we use the matrix equalities again so as to obtain
\begin{align}
&\qquad \frac{1}{2}\log\det(I_{2n_r}+\sunderline{H}_r\sunderline{K}_{\mathbf{X}}\sunderline{H}_r^H)-\frac{1}{2}\log\det(I_{2n_e}+\sunderline{H}_e\sunderline{K}_{\mathbf{X}}\sunderline{H}_e^H) \nonumber\\
&\leq  \log\det\left(I_{n_r+n_t}+
\begin{bmatrix}
\Delta^\frac{1}{2}&0\\0&H_e
\end{bmatrix}
\begin{bmatrix}
K_\mathbf{X}&K_\mathbf{X}\\K_\mathbf{X}&K_\mathbf{X}
\end{bmatrix}
\begin{bmatrix}
\Delta^\frac{H}{2}&0\\0&H_e^H
\end{bmatrix}\right)-\log\det\left(I_{n_e}+H_e K_\mathbf{X}H_e^H\right)\nonumber\\
&\overset{(a)}{=}\log\det\left(I_{n_r+n_t}+
\begin{bmatrix}
\Delta^\frac{1}{2}\\H_e
\end{bmatrix}
K_\mathbf{X}
\begin{bmatrix}
\Delta^\frac{H}{2}&H_e^H
\end{bmatrix}\right)-\log\det\left(I_{n_e}+H_e K_\mathbf{X}H_e^H\right)\nonumber\\
&\overset{(b)}{=}\log\det\left(I_{n_t}+
\begin{bmatrix}
\Delta^\frac{H}{2}&H_e^H
\end{bmatrix}
\begin{bmatrix}
\Delta^\frac{1}{2}\\H_e
\end{bmatrix}
K_\mathbf{X}
\right)-\log\det\left(I_{n_e}+H_e K_\mathbf{X}H_e^H\right)\nonumber\\
&=\log\det(I_{n_t}+ (\Delta+H_e^H H_e) K_{\mathbf{X}})-\log\det(I_{n_e}+H_e K_{\mathbf{X}} H_e^H)\nonumber\\
&\overset{(c)}{=}\log\det(I_{n_t}+ H_r^H H_r K_{\mathbf{X}})-\log\det(I_{n_e}+H_e K_{\mathbf{X}} H_e^H)\nonumber\\
&\overset{(b)}{=}\log\det(I_{n_r}+H_r K_{\mathbf{X}} H_r^H)-\log\det(I_{n_e}+H_e K_{\mathbf{X}} H_e^H),
\end{align}
for any $K_\mathbf{X} \in \sunderline{K}$. In the proof, (a) follows by \eqref{eq:m2}, (b) follows by \eqref{eq:m1} and (c) follows by \eqref{eq:delta}. As a result, if $\mathbf{X}$ is Gaussian with a given covariance matrix $K_\mathbf{X}$, its secrecy rate is maximized if $\mathbf{X}$ is proper, from which we get
\begin{align}\label{eq:20}
C_g\leq C_s.
\end{align}
\eqref{eq:20} combined with \eqref{eq:9} can yield 
\begin{align}
C_g=C_s.
\end{align}
We have proven that proper signals are capacity achieving for the degraded complex MIMO wiretap channel. 
\begin{remark}
From the perspective of mathematical optimization, we have
\begin{align}\label{eq:22}
&\mathop{\max}_{\sunderline{K}_{\mathbf{X}}\in \mathcal{\sunderline{K}}_g}[ \frac{1}{2}\log\det(I_{2n_t}+(\sunderline{\Delta}+\sunderline{H}_e^H\sunderline{H}_e)\sunderline{K}_{\mathbf{X}})-\frac{1}{2}\log\det(I_{2n_e}+\sunderline{H}_e\sunderline{K}_{\mathbf{X}}\sunderline{H}_e^H)]\nonumber\\
=&\mathop{\max}_{\sunderline{K}_{\mathbf{X}}\in \mathcal{\sunderline{K}}_p}[ \frac{1}{2}\log\det(I_{2n_t}+(\sunderline{\Delta}+\sunderline{H}_e^H\sunderline{H}_e)\sunderline{K}_{\mathbf{X}})-\frac{1}{2}\log\det(I_{2n_e}+\sunderline{H}_e\sunderline{K}_{\mathbf{X}}\sunderline{H}_e^H)],
\end{align}
which is more than just the equal optimal value of these two optimization problems with the same objective function and different feasible set. In fact, since the objective function is convex with an unique optimal point in set $\mathcal{\sunderline{K}}_g$, we have proven that the optimal $\sunderline{K}_\mathbf{X}$ can be found in the subset $\mathcal{\sunderline{K}}_p \subseteq \mathcal{\sunderline{K}}_g$.
\end{remark}
\end{proof}

\section{Optimality of the Proper Signal for the General Complex MIMO WTC}\label{sec4}
In this section, we rely on the min-max equivalent reformulation of the secrecy capacity and UDL decomposition in \cite{oggierSecrecyCapacityMIMO2011}. With these two skills, we can utilize the result of the degraded WTC in the previous section so as to prove the optimality of the proper signal for the general MIMO WTC.

\begin{lemma}
The secrecy capacity of a general MIMO WTC can be equivalently expressed in the form of a min-max optimization problem as
\begin{align}
    C_g=\mathop{\min}_{\sunderline{Q} \in \mathcal{\sunderline{Q}}}\mathop{\max}_{\sunderline{K}_{\mathbf{X}}\in \mathcal{\sunderline{K}}} \left[ \frac{1}{2}\log\det(I_{2(n_r+n_e)}+\sunderline{Q}^{-1}\sunderline{H}\sunderline{K}_{\mathbf{X}}\sunderline{H}^H)-\frac{1}{2}\log\det(I_{2n_e}+\sunderline{H}_e\sunderline{K}_{\mathbf{X}}\sunderline{H}_e^H) \right],
\end{align}
where $\sunderline{H}=[\sunderline{H}_r^T,\sunderline{H}_e^T]^T$ and the feasible set $\mathcal{\sunderline{Q}}$ is defined as
\begin{align}
\mathcal{\sunderline{Q}}=\{\sunderline{Q}|\sunderline{Q}=\begin{bmatrix}I_{2n_r}&\sunderline{A}\\\sunderline{A}^H&I_{2n_e}\end{bmatrix}\succ 0\},
\end{align}
which is the covariance matrix defined as
\begin{align}
\sunderline{Q}=\mathbb{E}\left\{\begin{bmatrix}N_r\\N_r^*\\N_e\\N_e^*\end{bmatrix}\begin{bmatrix}N_r\\N_r^*\\N_e\\N_e^*\end{bmatrix}^H\right\},
\end{align}
so $\sunderline{A}$ is the correlation between augmented noises $\sunderline{N}_r=\begin{bmatrix}N_r^H&N_r^T\end{bmatrix}^H$ and $\sunderline{N}_e=\begin{bmatrix}N_e^H&N_e^T\end{bmatrix}^H$ as
\begin{align}
\sunderline{A}=\mathbb{E}\{\sunderline{N}_r \sunderline{N}_e^H\}=\begin{bmatrix}
\mathbb{E}\{N_r N_e^H\}&\mathbb{E}\{N_r N_e^T\}\\\mathbb{E}\{N_r^* N_e^H\}&\mathbb{E}\{N_r^* N_e^T\}
\end{bmatrix}\triangleq\begin{bmatrix}
A&B\\B*&A^*
\end{bmatrix}.
\end{align}
\end{lemma}
\begin{proof}(Sketch)
The proof is similar to the proof of Theorem \ref{thm:1}. Firstly, in \cite{oggierSecrecyCapacityMIMO2011} Oggier et al. give the min-max reformulation for \eqref{eq:4} as
\begin{align}\label{eq:27}
C_p=\mathop{\min}_{Q \in \mathcal{Q}}\mathop{\max}_{K_\mathbf{x} \in \mathcal{K}}\left[ \log\det(I_{n_r+n_e}+Q^{-1}H K_{\mathbf{X}}H^H)-\log\det(I_{n_e}+H_e K_{\mathbf{X}}H_e^H) \right].
\end{align}
where $H=[H_r^T,H_e^T]^T$ and the feasible set $\mathcal{Q}$ is defined as
\begin{align}
\mathcal{Q}=\{Q|Q=\begin{bmatrix}I_{n_r}&A\\A^H&I_{n_e}\end{bmatrix}\succ 0\},
\end{align}
which is the covariance matrix defined as
\begin{align}
Q=\mathbb{E}\left\{\begin{bmatrix}N_r\\N_e\end{bmatrix}\begin{bmatrix}N_r\\N_e\end{bmatrix}^H\right\},
\end{align}
so $A$ is the correlation between noises $N_r$ and $N_e$ as $A=\mathbb{E}\{N_r N_e^H\}$.
Then we can find the corresponding result for real signals. Since we can view the general complex signal as a real signal composed of the real part and imaginary part of the original complex signal, the min-max reformulation for the general complex WTC can be obtained.  
\end{proof}
An UDL factorization method for $\sunderline{Q}$ was put forward in \cite[equation (55)]{oggierSecrecyCapacityMIMO2011}:
\begin{align}
\begin{bmatrix}I_{2n_r}&\sunderline{A}\\\sunderline{A}^H&I_{2n_e}\end{bmatrix}=\begin{bmatrix}I_{2n_r}&\sunderline{A}\\0&I_{2n_e}\end{bmatrix}
\begin{bmatrix}I_{2n_r}-\sunderline{A}\sunderline{A}^H&0\\0&I_{2n_e}\end{bmatrix}
\begin{bmatrix}I_{2n_r}&0\\\sunderline{A}^H&I_{2n_e}\end{bmatrix},
\end{align}
so that
\begin{align}
\begin{bmatrix}I_{2n_r}&\sunderline{A}\\\sunderline{A}^H&I_{2n_e}\end{bmatrix}^{-1}=\begin{bmatrix}I_{2n_r}&0\\-\sunderline{A}^H&I_{2n_e}\end{bmatrix}
\begin{bmatrix}(I_{2n_r}-\sunderline{A}\sunderline{A}^H)^{-1}&0\\0&I_{2n_e}\end{bmatrix}
\begin{bmatrix}I_{2n_r}&\sunderline{A}\\0&I_{2n_e}\end{bmatrix},
\end{align}
from which we get
\begin{align}\label{eq:32}
    C_g&=\mathop{\min}_{\sunderline{A} \in \mathcal{\sunderline{A}}}\mathop{\max}_{\sunderline{K}_{\mathbf{X}}\in \mathcal{\sunderline{K}}}C_g(\sunderline{A},\sunderline{K}_{\sunderline{K}_\mathbf{X}})\nonumber\\
    &=\mathop{\min}_{\sunderline{A} \in \mathcal{\sunderline{A}}}\mathop{\max}_{\sunderline{K}_{\mathbf{X}}\in \mathcal{\sunderline{K}}} [ \frac{1}{2}\log\det[I_{2n_t}+((\sunderline{H}_r^H-\sunderline{H}_e^H \sunderline{A}^H)(I_{2n_r}-\sunderline{A}\sunderline{A}^H)^{-1}(\sunderline{H}_r-\sunderline{A}\sunderline{H}_e)+\sunderline{H}_e^H\sunderline{H}_e)\sunderline{K}_{\mathbf{X}}]\nonumber\\
    &\quad-\frac{1}{2}\log\det(I_{2n_e}+\sunderline{H}_e\sunderline{K}_{\mathbf{X}}\sunderline{H}_e^H) ],
\end{align}
where the feasible set $\mathcal{\sunderline{A}}$ is defined as
\begin{align}
\mathcal{\sunderline{A}}=\{\sunderline{A}|I_{2n_r}-\sunderline{A}\sunderline{A}^H \succ 0\}.
\end{align}
Since the objective function in \eqref{eq:32} is convex in $\sunderline{A}$ and concave in $\sunderline{K}_{\mathbf{X}}$, the min-max is equal to the max-min and the optimum is a saddle point \cite[Proposition 5]{oggierSecrecyCapacityMIMO2011}.

If the feasible set $\mathcal{\sunderline{A}}$ is contracted to a smaller set $\mathcal{\sunderline{A}}'$ as
\begin{align}
\mathcal{\sunderline{A}}'=\left\{\sunderline{A}=\begin{bmatrix}
A&0\\0&A^*
\end{bmatrix}|I_{2n_r}-\sunderline{A}\sunderline{A}^H \succ 0\right\},
\end{align}
then the minimization in the set $\mathcal{\sunderline{A}}'$ will yield an upper bound of $C_g$. Denote the optimal $\sunderline{K}_\mathbf{X}$ by $\sunderline{K}_\mathbf{X}^*$, then we have
\begin{align}
C_g&=\mathop{\max}_{\sunderline{K}_{\mathbf{X}}\in \mathcal{\sunderline{K}}}\mathop{\min}_{\sunderline{A} \in \mathcal{\sunderline{A}}}C_g(\sunderline{A},\sunderline{K}_\mathbf{X})\nonumber\\
&=\mathop{\min}_{\sunderline{A} \in \mathcal{\sunderline{A}}}C_g(\sunderline{A},\sunderline{K}_\mathbf{X}^*)\nonumber\\
&\leq \mathop{\min}_{\sunderline{A} \in \mathcal{\sunderline{A}}'}C_g(\sunderline{A},\sunderline{K}_\mathbf{X}^*)\nonumber\\
&\overset{(a)}{\leq} \mathop{\max}_{\sunderline{K}_{\mathbf{X}}\in \mathcal{\sunderline{K}}}\mathop{\min}_{\sunderline{A} \in \mathcal{\sunderline{A}}'}C_g(\sunderline{A},\sunderline{K}_\mathbf{X}),
\end{align}
where (a) follows since contracting the feasible set $\mathcal{\sunderline{A}}$ to $\mathcal{\sunderline{A}}'$ may lead to a new optimal $\sunderline{K}_\mathbf{X}$ rather than $\sunderline{K}_\mathbf{X}^*$. Set $\mathcal{\sunderline{A}}$ is still a convex set, so the max-min is equal to min-max. It means
\begin{align}\label{eq:mm}
C_g& \leq \mathop{\min}_{\sunderline{A} \in \mathcal{\sunderline{A}}'}\mathop{\max}_{\sunderline{K}_{\mathbf{X}}\in \mathcal{\sunderline{K}}} C_g(\sunderline{A},\sunderline{K}_\mathbf{X})\nonumber\\
&=\mathop{\min}_{\sunderline{A} \in \mathcal{\sunderline{A}}'}\mathop{\max}_{\sunderline{K}_{\mathbf{X}}\in \mathcal{\sunderline{K}}} [ \frac{1}{2}\log\det[I_{2n_t}+((\sunderline{H}_r^H-\sunderline{H}_e^H \sunderline{A}^H)(I-\sunderline{A}\sunderline{A}^H)^{-1}(\sunderline{H}_r-\sunderline{A}\sunderline{H}_e)+\sunderline{H}_e^H\sunderline{H}_e)\sunderline{K}_{\mathbf{X}}]\nonumber\\
&\quad-\frac{1}{2}\log\det(I_{2n_e}+\sunderline{H}_e\sunderline{K}_{\mathbf{X}}\sunderline{H}_e^H)].
\end{align}
As each of the $\sunderline{A} \in \mathcal{\sunderline{A}}'$ is a block diagonal matrix so
\begin{align}\label{eq:d}
&\quad(\sunderline{H}_r^H-\sunderline{H}_e^H \sunderline{A}^H)(I-\sunderline{A}\sunderline{A}^H)^{-1}(\sunderline{H}_r-\sunderline{A}\sunderline{H}_e)\nonumber\\
&=\begin{bmatrix}
(H_r^H-H_e^HA^H)(I_{n_R}-AA^H)^{-1}(H_r-AH_e)&0\\
0&(H_r^T-H_e^TA^T)(I_{n_R}-A^*A^T)^{-1}(H_r^*-A^*H_e^*)
\end{bmatrix}\nonumber\\
&=\sunderline{\Delta}
\end{align}
is also a block diagonal matrix. In addition, \eqref{eq:d} is positive definite, hence we can denote it by $\sunderline{\Delta}$ with the form of
\begin{align}
\sunderline{\Delta}=\begin{bmatrix}
\Delta&0\\
0&\Delta^*
\end{bmatrix},
\end{align}
where $\sunderline{\Delta}$ and $\Delta$ are both positive definite with the existence of $\sunderline{\Delta}^{\frac{1}{2}}$ and $\Delta^{\frac{1}{2}}$. Now we can use the result of the degraded WTCs and get the following theorem.
\begin{theorem}
The secrecy capacity of general complex WTC is achieved when the signal is proper.
\end{theorem}
\begin{proof}
Since $I_{2n_r}-\sunderline{A}\sunderline{A}^H \succ 0$ is equivalent to $I_{n_r}-AA^H \succ 0$, we define $\mathcal{A}$ as
\begin{align}
\mathcal{A}=\{A|I_{n_r}-AA^H \succ 0\}.
\end{align}
Then we have
\begin{align}
C_g&\overset{(a)}{\leq}\mathop{\min}_{\sunderline{A} \in \mathcal{\sunderline{A}}'}\mathop{\max}_{\sunderline{K}_{\mathbf{X}}\in \mathcal{\sunderline{K}}_g} [ \frac{1}{2}\log\det[I_{2n_t}+((\sunderline{H}_r^H-\sunderline{H}_e^H \sunderline{A}^H)(I-\sunderline{A}\sunderline{A}^H)^{-1}(\sunderline{H}_r-\sunderline{A}\sunderline{H}_e)+\sunderline{H}_e^H\sunderline{H}_e)\sunderline{K}_{\mathbf{X}}]\nonumber\\
    &\quad-\frac{1}{2}\log\det(I_{2n_e}+\sunderline{H}_e\sunderline{K}_{\mathbf{X}}\sunderline{H}_e^H) ]\nonumber\\
&\overset{(b)}{=}\mathop{\min}_{\sunderline{A} \in \mathcal{\sunderline{A}}'}\mathop{\max}_{\sunderline{K}_{\mathbf{X}}\in \mathcal{\sunderline{K}}_g} \left[ \frac{1}{2}\log\det(I_{2n_t}+(\sunderline{\Delta}+\sunderline{H}_e^H\sunderline{H}_e)\sunderline{K}_{\mathbf{X}})-\frac{1}{2}\log\det(I_{2n_e}+\sunderline{H}_e\sunderline{K}_{\mathbf{X}}\sunderline{H}_e^H) \right]\nonumber\\
&\overset{(c)}{=}\mathop{\min}_{\sunderline{A} \in \mathcal{\sunderline{A}}'}\mathop{\max}_{\sunderline{K}_{\mathbf{X}}\in \mathcal{\sunderline{K}}_p} \left[ \frac{1}{2}\log\det(I_{2n_t}+(\sunderline{\Delta}+\sunderline{H}_e^H\sunderline{H}_e)\sunderline{K}_{\mathbf{X}})-\frac{1}{2}\log\det(I_{2n_e}+\sunderline{H}_e\sunderline{K}_{\mathbf{X}}\sunderline{H}_e^H) \right]\nonumber\\
&=\mathop{\min}_{A\in\mathcal{A}}\mathop{\max}_{K_{\mathbf{X}}\in \mathcal{K}} \left[\log\det(I_{n_t}+(\Delta+H_e^H H_e)K_{\mathbf{X}})-\log\det(I_{n_e}+H_e K_{\mathbf{X}}H_e^H) \right]\nonumber\\
&=\mathop{\min}_{Q \in \mathcal{Q}}\mathop{\max}_{K_\mathbf{x} \in \mathcal{K}}\left[ \log\det(I_{n_t}+H^HQ^{-1}H K_{\mathbf{X}})-\log\det(I_{n_e}+H_e K_{\mathbf{X}}H_e^H) \right]\nonumber\\
&=\mathop{\min}_{Q \in \mathcal{Q}}\mathop{\max}_{K_\mathbf{x} \in \mathcal{K}}\left[ \log\det(I_{n_r+n_e}+Q^{-1}H K_{\mathbf{X}}H^H)-\log\det(I_{n_e}+H_e K_{\mathbf{X}}H_e^H) \right]\nonumber\\
&\overset{(d)}{=} C_p,
\end{align}
where (a) follows by \eqref{eq:mm}, (b) follows by \eqref{eq:d}, (c) follows since Remark 1 is valid for any $\sunderline{A} \in \mathcal{\sunderline{A}}'$, so we have shown $C_g \leq C_p$ holds for the WTC, and (d) follows by \eqref{eq:27}. In addition, $C_p\leq C_g$ always holds as the proper signal is a special case of the general complex signal, so we have $C_g=C_p$, which finishes the proof.
\end{proof}
\section{Numerical Results} \label{sec5}
The following part of this paper moves on to compare the maximum achievable rate of the proper signal and general complex signal under the sum power constraint $P$ in the MIMO complex WTC. The channel matrices for the legitimate user and eavesdropper are
\begin{align}
H_r=\begin{bmatrix}
1.8+0.2i&0.8\\
1.5-0.4i&-0.8+1.1i
\end{bmatrix},
H_e=\begin{bmatrix}
0.9-0.7i&-1.2+1.4i\\
-0.3&-1.1-0.3i
\end{bmatrix},
\end{align}
which are generated randomly under the condition that $H_r^HH_r-H_e^HH_e$ has at least one positive eigenvalue to ensure the positive secrecy capacity. The eigenvalues are -2.6117 and 4.7017 for this case. Since the noises are assumed to be proper Gaussian with zero mean vectors and identical covariance matrices $I_2$, so the SNR is defined as $\SNR=10lg\frac{P}{2}$. In Fig.\ref{fig_6db} and Fig.\ref{fig_12db}, the convergence rate for $\SNR=6dB$ and $\SNR=12dB$ are depicted.
\begin{figure}[!t]
\centering
\includegraphics[width=0.7\textwidth]{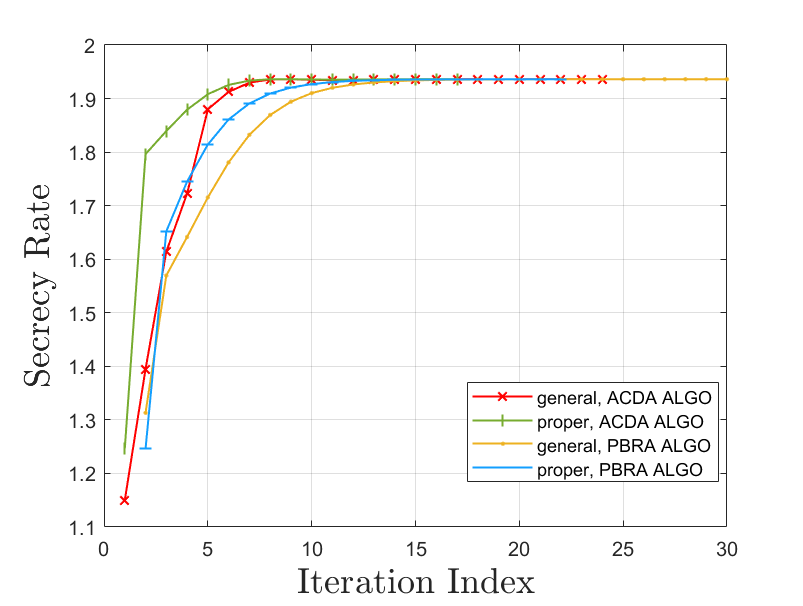}
\caption{Convergence rate comparison for the proper signal and general complex signal in the MIMO WTC, SNR=6dB.}
\label{fig_6db}
\end{figure}
\begin{figure}[!t]
\centering
\includegraphics[width=0.7\textwidth]{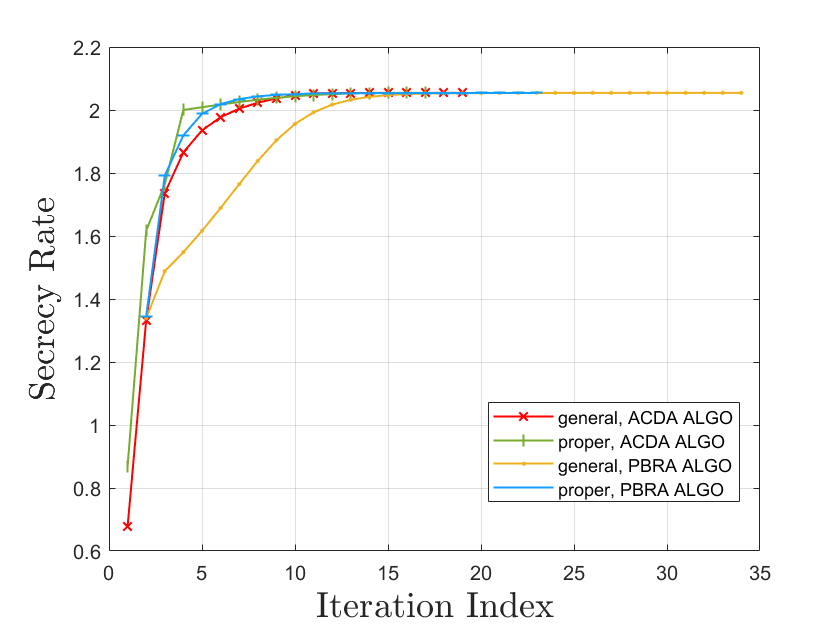}
\caption{Convergence rate comparison for the proper signal and general complex signal in the MIMO WTC, SNR=12dB.}
\label{fig_12db}
\end{figure}

The red curve and the yellow curve are the convergence rate for the general complex signal while the the green curve and the blue curve are the convergence for the proper signal. We used two different iterative algorithms to derive the secrecy rate. One is the ACDA, which is locally optimal, and the other is the PBRA, which is globally optimal \cite{mukherjeeSecrecyCapacityMIMO2021}. Both of the two algorithms are designed for the proper signal but we can modify them for the general complex signal. For each iteration curve, the initial point is chosen randomly so their initial secrecy rates are different. However, it can be observed that both proper signal and general complex signal achieve roughly the same maximum secrecy rate after several steps of iteration on both of the two figures. To show the iteration results more clearly, we report in Table \ref{tab1} the specific results of convergence, where the stopping criterion for ACDA and PBRA is when the increase is less than $10^{-5}$. The result of convergence in the table are the same with the error of less than $10^{-4}$. Therefore, the numerical results coincide with our results that the proper signal is capacity achieving in the MIMO complex WTC.
\begin{table}[t!]
\begin{center}
\begin{tabular}{ |p{4cm}||p{2.5cm}|p{2.5cm}|  }
 \hline
 \multicolumn{3}{|c|}{Maximum Achievable Rate} \\
 \hline
 signal and algorithm& $\SNR=6dB$&$\SNR=12dB$\\
 \hline
 general, ACDA   & 1.93626    &2.05447\\
 proper, ACDA &   1.93624  & 2.05446\\
 general, PBRA &1.93613 & 2.05440\\
 proper, PBRA    &1.93606 & 2.05446\\
 \hline
\end{tabular}
\end{center}
\caption{Comparison of the Iteration Result}
\label{tab1}
\end{table}
\section{Conclusion}\label{sec6}
In conclusion, the findings of this study can be understand as a complement to the research on the secrecy capacity of complex WTCs. We have demonstrated the correctness to focus on the proper signals in the complex WTC because the secrecy capacity is achieved only if the transmitted signal is proper Gaussian. It is intuitively assumed that the transmitted signal is proper in previous works about WTC, probably because the WTC can be viewed as a point-to-point channel plus a eavesdropper and the proper signal is capacity achieving for the point-to-point channel. However, the WTC is actually a BC with a confidential message and the existence of the eavesdropper add an "interference" to the secrecy capacity of the WTC, so the WTC is different from the point-to-point channel. Thus, the technique for the point-to-point channel is no longer applicable, hence we propose a new determinant function to fulfill the proof of the optimality of the proper signal in the MIMO complex WTC.


%

\appendices
\section{}\label{appendixa}
Oggier et al. \cite{oggierSecrecyCapacityMIMO2011} and Khisti et al.  \cite{khistiSecureTransmissionMultiple2010a} computed the secrecy capacity of the multiple antenna WTC with proper signals. Nevertheless, the proof can be transplanted to the situation with real signals without difficulty because all the related mathematical procedures in complex domain have counterparts in real domain. Therefore, the secrecy capacity of real MIMO wiretap channel is achieved when the real signal is Gaussian as
\begin{equation}
    C_s=\mathop{\max}_{K_{\mathbf{X}} \succeq 0, \tr(K_{\mathbf{X}})  \leq P} \left[ \frac{1}{2}\log\det(I_{n_r}+H_r K_{\mathbf{X}} H_r^T)-\frac{1}{2}\log\det(I_{n_e}+H_e K_{\mathbf{X}} H_e^T) \right],
\end{equation}
where $K_{\mathbf{X}}=\mathbb{E}\{\mathbf{X}\mathbf{X}^T\}$ and $H_r,H_e$ are real matrices\cite{tieliuNoteSecrecyCapacity2009}.

As for general complex signals, a complex signal $\mathbf{X}=\mathfrak{R}\{\mathbf{X}\}+i\mathfrak{I}\{\mathbf{X}\}\in \mathbb{C}^{n \times 1}$ can be regarded as a joint real signal $\soverline{\mathbf{X}}=(\mathfrak{R}\{\mathbf{X}\}^T,\mathfrak{I}\{\mathbf{X}\}^T)^T \in \mathbb{R}^{2n \times 1}$ and its information measures will keep the same because information measures such as the entropy are decided by the probability distribution of the random vector rather than the specific form of the random vector. Accordingly, \eqref{eq:1} can be written as
\begin{subequations}
\begin{align}
\underbrace{\begin{bmatrix}
    \sqrt{2}\mathfrak{R}\{\mathbf{Y}_r\}\\
    \sqrt{2}\mathfrak{I}\{\mathbf{Y}_r\}
\end{bmatrix}}_{\soverline{\mathbf{Y}}_r} = 
\underbrace{\begin{bmatrix}
    \sqrt{2}\mathfrak{R}\{H_r\}&-\sqrt{2}\mathfrak{I}\{H_r\}\\
    \sqrt{2}\mathfrak{I}\{H_r\}&\sqrt{2}\mathfrak{R}\{H_r\}
\end{bmatrix}}_{\soverline{H}_r}
\underbrace{\begin{bmatrix}
    \mathfrak{R}\{\mathbf{X}\}\\
    \mathfrak{I}\{\mathbf{X}\}
\end{bmatrix}}_{\soverline{\mathbf{X}}}   +
\underbrace{\begin{bmatrix}
\sqrt{2}\mathfrak{R}\{\mathbf{N}_r\}\\
\sqrt{2}\mathfrak{I}\{\mathbf{N}_r\}
\end{bmatrix}}_{\soverline{\mathbf{N}}_r}, \\\nonumber \\
\underbrace{\begin{bmatrix}
    \sqrt{2}\mathfrak{R}\{\mathbf{Y}_e\}\\
    \sqrt{2}\mathfrak{I}\{\mathbf{Y}_e\}
\end{bmatrix}}_{\soverline{\mathbf{Y}}_e} = 
\underbrace{\begin{bmatrix}
    \sqrt{2}\mathfrak{R}\{H_e\}&-\sqrt{2}\mathfrak{I}\{H_e\}\\
    \sqrt{2}\mathfrak{I}\{H_e\}&\sqrt{2}\mathfrak{R}\{H_e\}
\end{bmatrix}}_{\soverline{H_e}} 
\underbrace{\begin{bmatrix}
    \mathfrak{R}\{\mathbf{X}\}\\
    \mathfrak{I}\{\mathbf{X}\}
\end{bmatrix}}_{\soverline{\mathbf{X}}}   +
\underbrace{\begin{bmatrix}
\sqrt{2}\mathfrak{R}\{\mathbf{N}_e\}\\
\sqrt{2}\mathfrak{I}\{\mathbf{N}_e\}
\end{bmatrix}}_{\soverline{\mathbf{N}}_e},
\end{align}
\end{subequations}
where the constant $\sqrt{2}$ is multiplied to make noises $\soverline{\mathbf{N}}_r$ and $\soverline{\mathbf{N}}_e$ with identity covariance matrix.

We denote the covariance of $\soverline{X}$ by $\soverline{K}$, then we have
\begin{align}
\soverline{K}_\mathbf{X}=\mathbb{E}\left\{\begin{bmatrix}    \mathfrak{R}\{\mathbf{X}\}\\
    \mathfrak{I}\{\mathbf{X}\}\end{bmatrix}\begin{bmatrix}    \mathfrak{R}\{\mathbf{X}\}\\
    \mathfrak{I}\{\mathbf{X}\}\end{bmatrix}^T\right\}.
\end{align}

Furthermore, we can characterize the capacity with general complex signals as
\begin{align}\label{eq:a4}
    C_g=C_s=\mathop{\max}_{\soverline{K}_{\mathbf{X}} \succeq 0, \tr(\soverline{K}_{\mathbf{X}})  \leq 2P} \left[ \frac{1}{2}\log\det(I_{2n_r}+\soverline{H}_r \soverline{K}_{\mathbf{X}} \soverline{H}_r^T)-\log\det(I_{2n_e}+\soverline{H}_e \soverline{K}_{\mathbf{X}} \soverline{H}_e^T) \right].
\end{align}

It was shown in \cite{schreierStatisticalSignalProcessing2010} that
\begin{align}\label{eq:a5}
\soverline{K}_\mathbf{X}=\frac{1}{4}M_{n_t}^H\sunderline{K}_\mathbf{X}M_{n_t},
\soverline{H}_r=\frac{\sqrt{2}}{2}M_{n_t}^H\sunderline{H}_r M_{n_t}
\end{align}
where
\begin{align}
M_{n_t}=\begin{bmatrix}
I_{n_t} & i I_n\\I_{n_t} & -i I_{n_t}
\end{bmatrix}
\end{align}
is unitary matrix of factor 2, i.e.
\begin{align} \label{eq:a6}
M_{n_t} M_{n_t}^H=M_{n_t}^H M_{n_t}=2I_{2n_t}.
\end{align}
Substituting from \eqref{eq:a5} and \eqref{eq:a6} in \eqref{eq:a4}, we can get 
\begin{align}
    C_g=\mathop{\max}_{\sunderline{K}_{\mathbf{X}}\in \mathcal{\sunderline{K}}} \left[ \frac{1}{2}\log\det(I_{2n_r}+\sunderline{H}_r\sunderline{K}_{\mathbf{X}}\sunderline{H}_r^H)-\frac{1}{2}\log\det(I_{2n_e}+\sunderline{H}_e\sunderline{K}_{\mathbf{X}}\sunderline{H}_e^H) \right],
\end{align}
which is the secrecy capacity of general complex wiretap channels and the optimal signal is complex Gaussian. 
\section{}\label{appendixb}
For any positive semidefinite matrix $K$, we can find a proper random vector $\mathbf{X}$ so that $K=\mathbb{E}\{\mathbf{X} \cdot \mathbf{X}^H\}$. We can partition $\mathbf{X}$ to four sub random vectors $\mathbf{X}_i: i=1,2,3,4$ as
\begin{align}
\mathbf{X}^H=\begin{bmatrix}\mathbf{X}_1^H & \mathbf{X}_2^H & \mathbf{X}_3^H & \mathbf{X}_4^H\end{bmatrix},
\end{align}
where the dimension of $\mathbf{X}_i$ coincides with the dimension of $\mathcal{S}_i$, i.e. $\dim(\mathbf{X}_i)=\dim(\mathcal{S}_i)$, then we have
\begin{subequations}\label{eq:appendixb1}
\begin{align}
    h(\mathbf{X}_1,\mathbf{X}_2,\mathbf{X}_3,\mathbf{X}_4)&=\log((\pi e)^k|K|),\\
    h(\mathbf{X}_2,\mathbf{X}_4)&=\log((\pi e)^{k_2+k_4}|K(\mathcal{S}_2 \cup \mathcal{S}_4))|),\\
    h(\mathbf{X}_1,\mathbf{X}_2)&=\log((\pi e)^{k_1+k_2}|K(\mathcal{S}_1 \cup \mathcal{S}_2))|),\\
    h(\mathbf{X}_3,\mathbf{X}_4)&=\log((\pi e)^{k_3+k_4}|K(\mathcal{S}_3 \cup \mathcal{S}_4))|),\\
    h(\mathbf{X}_2)&=\log((\pi e)^{k_2}|K(\mathcal{S}_2)|),\\
    h(\mathbf{X}_4)&=\log((\pi e)^{k_4}|K(\mathcal{S}_4)|).    
\end{align}
\end{subequations}
Now consider the conditional entropy inequality
\begin{align}\label{eq:appendixb2}
    h(\mathbf{X}_1,\mathbf{X}_3|\mathbf{X}_2,\mathbf{X}_4)&=h(\mathbf{X}_1|\mathbf{X}_2,\mathbf{X}_4)+h(\mathbf{X}_3|\mathbf{X}_1,\mathbf{X}_2,\mathbf{X}_4) \nonumber\\
    &\leq h(\mathbf{X}_1|\mathbf{X}_2)+h(\mathbf{X}_3|\mathbf{X}_4).
\end{align}
Substituting \eqref{eq:appendixb1} into \eqref{eq:appendixb2}, we obtain the desired result.

%
\bibliographystyle{IEEEtran}
\bibliography{ref}

%




\end{document}